\documentclass[12pt]{amsproc}

\usepackage[margin=1in]{geometry}
\usepackage{concmath}
\usepackage{color}

\usepackage{mathtools}
\allowdisplaybreaks[1]
\mathtoolsset{showonlyrefs=true}
\usepackage{dsfont}
\usepackage{amssymb}
\usepackage{enumerate}
\usepackage{setspace}

\newcommand{\novec}{{}}
\renewcommand{\matrix}{{}}

\newcommand{\mL}{\mathcal{L}}
\newcommand{\mN}{\mathcal{N}}

\newcommand{\mS}{\mathcal{S}}

\newtheorem{theorem}{Theorem}[section]
\newtheorem{lemma}[theorem]{Lemma}

\title[Mechanism Design for Unicast Service]{Generalized Proportional Allocation Mechanism Design for Unicast Service on the Internet}

\author[Abhinav Sinha and Achilleas Anastasopoulos]{Abhinav Sinha and Achilleas Anastasopoulos \\
EECS Department, University of Michigan, Ann Arbor, MI 48109 \\
\texttt{\{absi,anastas\}@umich.edu}}

\begin{document}
\begin{spacing}{1.2}

\begin{abstract}
In this report we  construct  two mechanisms that fully implement social welfare maximising allocation in Nash equilibria for the case of a single infinitely divisible good subject to multiple inequality constraints. The first mechanism achieves weak budget balance, while the second is an extension of the first, and achieves strong budget balance.
One important application of this mechanism is unicast service on the Internet where a network operator wishes to allocate rates among strategic users in such a way that maximise overall user satisfaction while respecting capacity constraints on every link in the network. The emphasis of this work is on full implementation, which means that all Nash equilibria of the induced game result in the optimal allocations of the centralized allocation problem.
\end{abstract}

\maketitle
\begin{center}
May 4, 2013
\end{center}


\section{Centralised Problem} \label{secCP}

Consider a set $ \mN = \{1,2,\ldots,N\} $, of $ N $ of Internet agents (an agent is considered a pair of source and destination users) that communicate over pre-specified routes on the Internet. Each agent $i\in\mN$, communicates at an information rate $x_i\in \mathbb{R}_{+}$ (where $ \mathbb{R}_{+}$ is the set of non-negative real numbers). Agent's valuation for an overall rate allocation
$ x=(x_i)_{i\in\mN} \in \mathbb{R}_{+}^N $, can be written as
\begin{equation}
\tilde{v}_i(x) = v_i(x_i) \qquad \forall~~i \in \mN
\end{equation}
where $ v_i: \mathbb{R}_{+}\rightarrow \mathbb{R}$, for all $ i \in \mN $, which indicates that
agent $i$'s satisfaction only depends on its own information rate allocation $x_i$.
Due to capacity constraints on the utilized links, allocation to agents is constrained
by a number of inequality constraints. It is assumed that although some agents may share information content (e.g., watch the same video stream), a separate data stream is transmitted for each agent.
This transmission technique is referred to as unicast service.

Each agent has a fixed pre-determined route. The route $ \mL_i $ for agent $ i $ is the set of links that agent $ i $ uses for his communication, and $ \mL = \cup_{i \in \mN} \mL_i $ is the set of all available links. We also define the sets of agents utilizing link $l\in\mL$ as $\mN^l = \{i \in \mN ~\vert~ l \in \mL_i \}$. Finally, for any agent $ i $ we denote $ L_i = \vert \mL_i \vert $ and for any link $ l $ we denote $ N^l = \vert \mN^l \vert $.

The network administrator is interested in maximizing the social welfare under the link capacity constraints. This centralized problem~\eqref{CP} is formally defined below.
\begin{alignat}{2}\label{CP}
            &\max_x \sum_{i \in \mN} v_i(x_i) \tag{CP}\\
\label{C1}
\text{s.t.} &\quad x_i \ge 0 \quad \forall~~ i \in \mN \tag{C$_1$} \\
\label{C2}
\text{and}  &\quad \sum_{j \in \mN^l} \alpha^l_j x_j \le c^l \quad  \forall~~ l \in \mL \tag{C$_2$} \\
\end{alignat}
Specifically, constraints \eqref{C2} are the inequality constraints on allocation, which as mentioned above, can be interpreted as capacity constraint for every link $ l \in \mL$, in the network.
In this interpretation $ \alpha^l_j $ would be representative of the QoS requirement of agent $ j $ combined with the specific architecture on link $ l $.
As an example, $\alpha^l_j=\frac{1}{R_j(1-\epsilon^l_{j})}$ for all links $l\in\mL_j$, where $\epsilon^l_{j}$ represents the packet error probability for link $l$ for a packet encoded with channel coding rate $R_j$.

\subsection{Assumptions} Our analysis would be done under the following assumptions.

\begin{enumerate}
\item[(A1)] For all agents, $ v_i(\cdot) \in \mathcal{V}_i $, where the sets $\mathcal{V}_i$ are arbitrary subsets of $\mathcal{V}_0$, the set of strictly increasing, strictly concave, twice differentiable functions $\mathbb{R}_{+}\rightarrow \mathbb{R}$ with continuous second derivative.

\item[(A2)] $ v_i^{\prime}(0) $ is finite $ \forall $ $ i \in \mN $. This also implies that $ v_i^{\prime}(x) $ is finite and bounded $ \forall $ $ i $ and $ \forall $ $ x $ since $ v_i $'s are concave.

\item[(A3)] There are at least two agents on each link i.e. $ N^l \ge 2 $ $ ~\forall $ $ l \in \mL $.

\item[(A4)] The optimal solution of \eqref{CP} has at least two non-zero components at each link i.e. if $ S(x) \coloneqq \{ i \in \mN ~\vert~ x_i > 0 \} $ and $ S^l(x) \coloneqq S(x) \cap \mN^l $ then we assume $ \vert S^l(x^{\star}) \vert \ge 2 $ $ ~\forall $ $ l \in \mL $. (where $ x^{\star} $ is the solution of \eqref{CP})
\end{enumerate}
In addition, the coefficients in \eqref{C2} are all strictly positive, i.e. $ \alpha^l_j > 0 $ $ \forall $ $ j \in \mN^l $, $ \forall $ $ l \in \mL $. Also, for well-posedness of the problem we take $ c^l > 0 $ $ ~\forall $ $ l \in \mL$.

Assumption (A1) is made in order for the centralized problem to have a unique solution and for this solution to be characterized by the KKT conditions.
(A2) is a mild technical assumption that is required in the proof of Lemma~\ref{lemexis}.
Assumption (A3) is  made in order to avoid situations where there is a link constraint involving only one agent.
Such case requires special handling in the design of the mechanism (since in such a case there is no contention at the link), and destructs from the basic idea
that we want to communicate.
Finally (A4) is related to (A3) and is made in order to simplify the exposition of the proposed mechanism,
without having to define corner cases that are of minor importance.
We will present our proposed mechanism in a concise way using assumption (A4) and generalisations will be discussed separately in Section \ref{secgen}. 

\subsection{Necessary and Sufficient Optimality conditions} We can write the following four KKT conditions, which are generally necessary, but in our case (due to all constraints being affine and strict concavity of $ v_i $) they will also be sufficient. For that we first define the Lagrangian
\begin{equation}
L(x,\lambda, \nu) = \sum_{i \in \mN} v_i(x_i) - \sum_{l \in \mL} \lambda_l \left (\sum_{j \in \mN^l} \alpha^l_j x_j - c^l \right ) + \sum_{i \in \mN} \nu_i x_i
\end{equation}
We will write KKT conditions without explicitly referring to $ \nu_i $'s and just using the fact that $ \nu_i^{\star} \ge 0 $ and $ \nu_i^{\star} x_i^{\star} = 0 $ $ ~~\forall~~i \in \mN$.
With the assumptions above, it's easy to see that the KKT conditions below will give rise to a unique $ x^{\star} $
as the optimiser for \eqref{CP}.

\underline{KKT conditions}:
\begin{enumerate}
\item[a)] \textsf{Primal Feasibility: }
\begin{equation}
x_i^{\star} \ge 0 \quad  \forall~~ i \in \mN  \quad  \text{and} \quad   \sum_{j \in \mN^l} \alpha^l_j x_j^{\star} \le c^l \quad  \forall~~ l \in \mL
\end{equation}

\item[b)] \textsf{Dual Feasibility: } $ \lambda^{\star}_l \ge 0~~ $ $ ~~\forall $ $ l \in \mL $

\item[c)] \textsf{Complimentary Slackness: }
\begin{equation}
\lambda^{\star}_l \left( \sum_{j \in \mN^l} \alpha^l_j x_j^{\star} - c^l \right) = 0 \quad \forall ~~ l \in \mL
\end{equation}

\item[d)] \textsf{Stationarity: }
\begin{gather}
v_i^{\prime}(x_i^{\star}) = \sum_{l \in \mL_i} \lambda^{\star}_l \alpha^l_i \quad \forall~~ i \in \mN \quad \text{s.t.} \quad x_i^{\star} > 0 \\
v_i^{\prime}(x_i^{\star}) \le \sum_{l \in \mL_i} \lambda^{\star}_l \alpha^l_i \quad \forall~~ i \in \mN \quad \text{s.t.} \quad x_i^{\star} = 0
\end{gather}

\end{enumerate}

\section{Different Formulations of the Centralised Problem}

The designer's task is to ensure the above optimum allocation is made. This clearly requires the knowledge of $ v_i $ even when constraints \eqref{C1} and \eqref{C2} are completely known. The premise of the problem, however, is that we are dealing with agents who are \textrm{\textit{strategic}} and for each of whom, their own valuation function $ v_i(\cdot) $ is their private information.
One way forward for the designer could be to simply ask each agent to report their private information and announce the solution of \eqref{CP}, with reported functions in place of $ v_i $, for allocation.
Apart from the fact that asking to report a function creates a practical communication problem,
the main problem with this is that the agents could report untruthfully and end up getting a strictly better
allocation (e.g., by reporting a $ v_i $ which has higher derivative than original at every point).
In mechanism design terminology, the allocation function arising out of \eqref{CP} isn't even partially
implementable\footnote{This can be deduced from the revelation principle.
Indeed if there was a mechanism that even partially implements
the allocation function arising out of \eqref{CP},
then there would exist also a truthful implementation. However, as
shown with the above example, such an implementation will always fail.}
Restricting ourselves to a certain class of utility functions (quasi-linear utilities),
provides the additional flexibility of penalising agents for reporting untruthfully by imposing taxes/subsidies.
In this way, another related problem is created which is implementable, and which is equivalent to \eqref{CP} as far as allocation is concerned. This leads us to the following additional
assumption about users' utilities
\begin{itemize}
\item[(A5)] All agents have {{quasi-linear}} utilities, i.e. the overall utility functions can be expressed as
\begin{equation} \label{EQUT}
u_i(x,t) = v_i(x_i) - t_i \qquad \forall~~i \in \mN
\end{equation}
where in addition to allocation we have introduced taxes $ t=(t_i)_{i\in\mN} \in \mathbb{R}^N $.
\end{itemize}
Note that under assumption (A5), agent $ i $ pays tax if $ t_i > 0 $ and receives a subsidy if $t_i<0$.
Taxes affect utilities linearly and overall utility itself is valuation after adjustment for taxes
(total monetary representation of one's state of happiness).

Because we talk about social welfare as our main objective, the centralised problem \eqref{CP} isn't complete until we fix who owns the good that is being allocated. Then one will have to further check whether including their welfare in the objective function changes the optimum allocation. As it turns out, under the assumption of \textrm{\textit{quasi-linear}} utilities and cost of providing the good being zero for the owner, optimum doesn't change even if we involve the seller's welfare. In this regard, there are two interesting ways of reformulating \eqref{CP}, as elaborated below.

\subsection{First Reformulation of CP: Weak budget balance}

We now introduce agent $ 0 $ as the owner of the good (called the seller). The seller doesn't have any costs for producing and providing the good, i.e. his valuation is the zero function. This could be interpreted as the good being already produced and ready to be provided, so those costs don't come into consideration for the seller as well as the designer. His utility is linear (since valuation is zero) and his revenue is the total tax paid by the agents, $\sum_{i\in\mN} t_i $.

We define centralised problem \eqref{CP1} as
\begin{gather}\label{CP1}
\max_{x,t} \sum_{i \in \mN} u_i(x,t) + \sum_{j \in \mN} t_j \tag{CP$_1$}\\
\text{s.t.} \quad \eqref{C1} \quad \text{and} \quad \eqref{C2} 
\end{gather}
where now, instead of just taking agent's valuations into account, we maximise the sum of their overall utilities, with the addition of seller's utility (which is only his revenue) - each agent pays a tax $ t_i $, all of which goes to the seller, who has no valuation and therefore has utility equal to sum of taxes.
Anticipating that a rational seller will only sell if his revenue is non-negative
we can add a weak budget balance (WBB) constraint, which states
\begin{equation} \label{wbb}
 \sum_{j\in \mN} t_j \ge 0 \tag{WBB}.
\end{equation}

\subsection{Second Reformulation of CP: Strong budget balance}

In this case, in contrast to \eqref{CP1}, there is no separate seller.
We can alternatively say that the agents are themselves the owners of the good and are only
looking to distribute the good (which they collectively own) in a way such that sum of utilities is maximised.
Therefore strong budget balance \eqref{sbb}  constraint is needed.
This means that for the system $ \mN $, no money has been introduced from the outside and the agents wish
that no excess money remains undistributed after the allocation either.

The new centralised problem \eqref{CP2} resulting from the above interpretation can be stated as
\begin{gather}\label{CP2}
\max_{x,t} \sum_{i \in \mN} u_i(x,t) \tag{CP$_2$}\\
\text{s.t.} \quad \eqref{C1} \quad \text{and} \quad \eqref{C2} \\
\label{sbb}
\text{and} \quad \sum_{i \in \mN} t_i = 0 \tag{SBB}.
\end{gather}

\subsection*{}
The two problems defined above will be shown to be equivalent to \eqref{CP} where since the original problem~\eqref{CP} did not involve taxes, we will talk of equivalence only in terms of optimum allocation, $ x^{\star} $. Note that due to different conditions on taxes in the two, two different mechanisms will be needed to implement them.

It is straightforward to see that \eqref{CP2} and \eqref{CP} are completely equivalent - due to constraint $ SBB $, the objective for \eqref{CP2} is independent of $ t $ and is exactly the same as objective for \eqref{CP}, with same remaining constraints. Now for \eqref{CP1} and \eqref{CP2}, since the constraints on $ x $ are the same in \eqref{CP1} and \eqref{CP2} and the $ x- $dependent part of the objective in \eqref{CP2} in independent of $ t $ and is the same as the objective of \eqref{CP1}, we can see that \eqref{CP1} and \eqref{CP2} are equivalent. The two equivalences above automatically give the third one i.e. \eqref{CP} and \eqref{CP1}. 

The above equivalences mean that not only will $ x^{\star} $ be the same, but also that the necessary and sufficient conditions describing it will be the same i.e. KKT conditions, for $ x^{\star} $ and $ \lambda^{\star} $, will be exactly the same for all three problems (additionally we will show $ WBB $ and $ SBB $ constraints to be satisfied in respective formulations). This fact will be used in Sections \ref{secmech}, \ref{secmechown} where the KKT conditions from Section \ref{secCP} will be treated as if they have been written for \eqref{CP1}, \eqref{CP2}, respectively.

In Section \ref{secmech}, we will present a mechanism that fully implements \eqref{CP1} in Nash Equilibria (NE), while in Section \ref{secmechown} we will modify our mechanism to fully implement \eqref{CP2} in NE.

\section{A Mechanism with Weak Budget Balance} \label{secmech}
In this section we refer to \eqref{CP1} as the centralised problem.
So we have all the agents in $ \mN $ plus the seller and social welfare is in terms of everyone's
utility (including seller's).

We will define a mechanism, in a way that doesn't require knowledge of $ v_i $, whose
game-form will have NE in pure strategies such that the allocation which corresponds to the
equilibria of the game-form is same across all equilibria and is equal to the unique optimiser
of \eqref{CP1}, $ x^{\star} $.
In addition, the mechanism will be such that everyone involved (including the seller) will be weakly
better-off at equilibrium than not participating at all.

\subsection{Information assumptions}

Assume that $ v_i(\cdot) $ is a private information of agent $ i $ and nobody else knows it\footnote{This
assumption is crucial because it raises the question of the validity of NE as a solution concept of the resulting
game, since that would require that all agents have complete information about everyone's utilities. We believe this
ia a serious problem in this entire line of research and that a Bayesian formulation would be more appropriate.
However, in this work we accept the justification--weak in our opinion--given by Reichelstein and Reiter
in \cite{reiter} and Groves and Ledyard in \cite{groves}.}.
Let $ \mathcal{I}_c $ be the set of common information between all agents, containing the information about full rationality of each agent.
Finally, let $ \mathcal{I}_d $ be the knowledge of the designer, containing the information about constraints \eqref{C1}, \eqref{C2}, the fact that $\mathcal{V}_i \subset \mathcal{V}_0,\ \forall i\in\mN$ and that the seller has $ 0 $ valuation.

\subsection{Mechanism}

Formally, we have a set of environments $ \mathcal{V} = \times_{i \in \mN} \mathcal{V}_i $.
We have seen from KKT, how each element of $\mathcal{V} $ can be mapped to an allocation $ x^{\star} $ which
maximises social welfare for that set of utilities.
The allocation $ x^{\star} $ achieves the maximum of \eqref{CP1}, and correspondingly any
tax $ t \in \mathbb{R}^{N} $  satisfying~\eqref{wbb} would do.

In our mechanism, the designer would define an action (message) space $\mS_i$ for each agent $i\in\mN$.
We denote $ \mS = \times_{i \in \mN} \mS_i $ the set of action profiles for all agents.
In addition the designer defines and announces the \textrm{\textit{contract}} $ h: \mS \rightarrow \mathbb{R}_{+}^N \times \mathbb{R}^N $ that maps every vector of messages received from the agents into an allocation vector and a tax vector. The designer would then ask every agent $ i \in \mN $ to choose a message from the set $ \mS_i $ based on which allocations (and taxes) would be made. The seller is not asked to take any action, so as far as strategic decision making is concerned, we don't need to consider him any further. It is implicit in our mechanism in this section that when the tax $ t $ is imposed, the seller gets revenue (or utility) of $ \sum_{i\in\mN} t_i $.

Specifically, the designer would ask each agent to report $s_i= (y_i, \novec{p}_i) $, $ \novec{p}_i = (p_i^l)_{l \in \mL_i}$.
This includes their demand for the good and the ``price'', for each constraint they are involved in, that they believe everyone should pay.
This means $ \mS_i = \mathbb{R}_{+} \times \mathbb{R}_{+}^{L_i} $.
For received messages $ s = (s_1, \ldots, s_N) = (\novec{y}, \matrix{P}) = (y_1, \ldots, y_N, \novec{p}_1 \ldots, \novec{p}_N ) $ the
contract $ h_i(s) = (h_{x,i}(s), h_{t,i}(s)) $ will be defined for each $ i \in \mN $ as follows.

If the received demand vector is $ y = (y_1, \ldots, y_N) = 0 $ the allocation is $ x= (x_1, \ldots, x_N) = 0 $.
Otherwise it is evaluated by first generating a scaling factor $r$ through
\begin{equation} \label{EQx0}
r = \min_{l \in \mL} r^l,
\end{equation}
where
\begin{align} \label{EQA1}
r^l &= \left\{
\begin{array}{ll}
\frac{c^l}{\sum_{j \in \mN^l} \alpha^l_j y_j }, & \quad \text{if }~ \vert S^l(y) \vert \ge 2 \\
\frac{c^l}{\sum_{j \in \mN^l} \alpha^l_j y_j } - f^l(y_i), & \quad \text{if }~ S^l(y) = \{i\} \\
+ \infty  & \quad \text{if }~ \vert S^l(y) \vert = 0
\end{array}
\right.
\end{align}
with
\begin{equation}
f^l(y_i) = \frac{c^l}{\alpha_{i}^ly_i (y_i + 1)}.
\end{equation}
Using these previously defined quantities, the allocation and taxes are
\begin{align} \label{EQx}
h_{x,i}(s) = x_i = r y_i 
\end{align}
\begin{gather}
\label{EQt}
h_{t,i}(s) = t_i = \sum_{l \in \mL_i} t_i^l \\
t_i^l =  x_i \alpha^l_i \bar{p}_{-i}^l  + (p_i^l -  \bar{p}_{-i}^l)^2  + \eta \, \bar{p}_{-i}^l \left( p_i^l -  \bar{p}_{-i}^l \right) \left(c^l - \sum_{j \in \mN^l} \alpha^l_jx_j \right),
\end{gather}
where $ \eta $ is a small enough positive constant (described in proof of Lemma~\ref{lemexis}) for any link $ l \in \mL_i $ we define $ \bar{p}_{-i}^l $ by
\begin{align} \label{EQpm}
\bar{p}_{-i}^l \coloneqq
\frac{1}{\vert \mN^l \backslash \{i\} \vert} \sum_{j \in \mN^l \backslash \{i\}} p_j^l = \frac{1}{N^l - 1} \sum_{j \in \mN^l \backslash \{i\}} p_j^l.
\end{align}
The quantity $ \bar{p}_{-i}^l $ is calculated by averaging the quoted prices for link $ l $ over all
agents other than $ i $ who use that link (note that due to assumption (A3) every link has at least 2 agents). 
The interpretation of prices $ p_j^l $ in this mechanism is closely related to agent $ j $'s willingness to pay for consuming resource on link $ l $.
Since we have the problem of information elicitation for each agent's type ($ v_i $), quoting of prices and demand is used as a way of eliciting $ v_i^{\prime}(x_i) $ by comparing it appropriately with prices. 

The quantity $ h_{x,i}(s) $ creates allocation by dilating/shrinking a given demand vector $ \novec{y} $ on to one of the hyperplanes defined
by the constraints in \eqref{C2}, specifically, that hyperplane for which the corresponding allocation is the closest to origin
(the allocation could also be at the intersection of multiple hyperplanes).
Another way to describe this is to say that $ h_{x,i} $ dilates/shrinks $ \novec{y} $ to the boundary of the feasible region defined by
\eqref{C2} constraints. Since all the $ \alpha^l_j $'s are positive, this means that all other constraints in \eqref{C2} are satisfied for the
allocation automatically (shown later).
Additionally, the separate definition for $ r^l $ when $ \vert S^l(y) \vert \leq 1 $ is to ensure (as it will be shown later) that there are no equilibria
with $ \vert S^l(y) \vert \leq 1 $. This is required since we are only dealing with achieving solutions\footnote{Note that for the given allocation function, $ \vert S^l(y) \vert  = 0, 1 $ is equivalent to $ \vert
S^l (x) \vert = 0, 1 $.}
 to \eqref{CP} which satisfy assumption (A4).

The mechanism gives rise to a one-shot game $ \mathcal{G} $, played by all the agents in $ \mN $,
where action sets are $ (\mS_i)_{i \in \mN} $ and utilities are given by
\begin{equation}
\hat{u}_i(s) = v_i(x_i) - t_i = v_i(h_{x,i}(s)) - h_{t,i}(s) \qquad \forall~~i \in \mN
\end{equation}
We will say that maximising social welfare for \eqref{CP1} has been \textrm{\textit{fully implemented in NE}},
if all outcomes (all possible NE) of this game  produce allocation $ x^{\star} $ and all agents in $ \mN $
plus the seller are better-off participating in the mechanism than opting out (getting $ 0 $ allocation and
taxes). The second property is known as \textrm{\textit{individual rationality}}.
%

\subsection{Results}
\begin{theorem}[Full Implementation] \label{thmain}
For game $ \mathcal{G} $, there is a unique allocation, $ x $, corresponding to all NE. Moreover, $ x = x^{\star} $, the maximiser of \eqref{CP}. In addition, individual rationality is satisfied for all agents and for the seller.
\end{theorem}

The theorem will be proved by a sequence of results, in which we will characterise all candidate NE of $ \mathcal{G} $ by necessary conditions until
we are left with only one family of NE candidates.
We will then show that $ \mathcal{G} $ has NE in pure strategies, and that all of them result in allocation $ x = x^{\star} $.
Finally, individual rationality will be checked.


\begin{lemma}[Primal Feasibility] \label{lemPF}
For any action profile $ s = (\novec{y}, \matrix{P}) $ of game $ \mathcal{G} $, constraints \eqref{C1} and \eqref{C2} are satisfied at the corresponding allocation.
\end{lemma}
\begin{proof}
Constraint \eqref{C1} is clearly always satisfied.
For $ y = 0 $, constraint \eqref{C2} is also clearly satisfied.
We will now show \eqref{C2} is satisfied for any $ y \ne 0 $.
In that case $r<+\infty$ (since there exists at least one link $q$ with $\vert S^q(y) \vert\geq 1$ and thus  $r^q<+\infty$).
Now, for any link $ l $, we have the following two cases.
If $ \vert S^l(y) \vert = 0 $ then the allocation to agents on that link is clearly zero (since $ x_i = r y_i $ and $y_i = 0$),
so \eqref{C2} for those links is  satisfied.
If $ \vert S^l(y) \vert \ge 1 $ we have
\begin{gather}
\sum_{j \in \mN^l} \alpha^l_j x_j = r \sum_{j \in \mN^l} \alpha^l_j y_j \le r^l \sum_{j \in \mN^l} \alpha^l_j y_j \le \frac{c^l}{\sum_{j \in \mN^l} \alpha^l_j y_j} \sum_{j \in \mN^l} \alpha^l_j y_j = c^l
\end{gather}
where the first inequality holds because $ r $ is the minimum of all $ r^l $'s. The second inequality will be equality if $ \vert S^l(y) \vert \ge 2 $ and will be strict only if $ \vert S^l(y) \vert = 1 $ (see second sub-case in eq. \eqref{EQA1}).
\end{proof}

Feasibility of allocation for action profiles is a direct consequence of using projections of demand $ y $ on to the feasible region.
Now we will prove that all agents using a link quote the same price for it, this is brought about by the 2nd tax term
$ \sum_{l \in \mL_i} (p_i^l - \bar{p}_{-i}^l)^2 $. This is a way of penalizing users with higher taxes just for quoting
a different price than average, at each link.

\begin{lemma} \label{lemcmp}
At any NE $ s = (\novec{y}, \matrix{P}) $ of $ \mathcal{G} $, price for any link, is same for all agents, i.e.,  $ {p}_i^l = p^l $ $~ \forall$ $ i \in \mN^l $, $ ~\forall $ $ l \in \mL $ (where we denote the common price at any link $ l $ by $ p^l $).
\end{lemma}
\begin{proof}
Suppose there exist a link $ q $ for which prices $ ({p}_i^q)_{i \in \mN^q} $ at equilibrium are not all equal.
Clearly then there is an agent $ j \in \mN^q $ for whom $ {p}_j^q > \bar{p}_{-j}^q $ (this can be seen from eq. \eqref{EQpm}).
We will show that this agent can deviate by only reducing price for link $ q $ and be strictly better off, thereby contradicting the equilibrium condition.

Take the deviation by agent $ j $, where $ {p^{\prime}_j}^q = \bar{p}_{-j}^q < {p}_j^q $ and his demand remains the same.
Then the difference between utilities for agent $ j $ after and before deviation would arise only because of change in taxes
(since allocations haven't changed for any agent) and moreover the difference would only arise from tax terms corresponding to
link $ q $ (refer to \eqref{EQt}). We'll have
\begin{gather}
\Delta \hat{u}_j = - ({p^{\prime}_j}^q -  \bar{p}_{-j}^q)^2 - \eta \, \bar{p}_{-j}^q({p^{\prime}_j}^q - \bar{p}_{-j}^q)(c^q - \sum_{k \in \mN^q} \alpha^q_k {x}_k ) \\
{} +  \eta \, \bar{p}_{-j}^q({p}_j^q - \bar{p}_{-j}^q)(c^q - \sum_{k \in \mN^q} \alpha^q_k {x}_k )   + ({p}_j^q - \bar{p}_{-j}^q)^2 \\
= -0 - 0 + \underbrace{({p}_j^q - \bar{p}_{-j}^q)}_{> 0} \big [ \eta \, \bar{p}_{-j}^q\underbrace{(c^q - \sum_{k \in \mN^q} \alpha^q_k {x}_k )}_{\ge 0 ~by~ Lemma~ \ref{lemPF}}   + \underbrace{({p}_j^q - \bar{p}_{-j}^q)}_{> 0} \big] > 0
\end{gather}
which shows that the above deviation is a profitable one.

Hence at equilibrium, for any link, the price quoted for that link by any user using that link is the same, we denote the common price vector by $ \novec{p} = (p^l)_{l \in \mL} $.
\end{proof}


Now that we established Lemma \ref{lemcmp}, we can talk in terms of the common price vector at equilibrium rather than different price vectors for all agents, in fact we can identify each NE candidate profile $ s = (\novec{y}, \matrix{P}) $ with $ s = (\novec{y}, \novec{p}) $, with $p=(p^l)_{l\in\mL}$. 

We will later see how $ \novec{p} $ will take the place of dual variables $ \lambda $  when we compare equilibrium conditions with KKT conditions,
hence we identify the following condition as dual feasibility.

\begin{lemma}[Dual Feasibility] \label{lemDF}
$ p^l \ge 0 $ $ \quad \forall \ l \in \mL $.
\end{lemma}
\begin{proof}
This is also by design, since any agent $ i $ is asked to select a price vector in $ \mathbb{R}_{+}^{L_i} $.
\end{proof}

Following is the property that solidifies the notion of prices as dual variables, since here we claim that inactive constraints do not contribute
to payment at equilibrium. This notion is very similar to the centralised problem, where if we know certain constraints to be
inactive at the optimum then the same problem without these constraints would be equivalent to the original.
The 3rd term in the tax function facilitates this by charging extra taxes for inactive constraints when the agent is quoting higher
prices than the average of remaining ones, thereby driving prices down.

\begin{lemma}[Complimentary Slackness] \label{lemCompSlac}
At any NE $ s = (\novec{y},\novec{p})$ of game $ \mathcal{G} $ with corresponding allocation $ x $ satisfies
\begin{equation}
p^l \left ( \sum_{k \in \mN^l} \alpha^l_k {x}_k - c^l \right) = 0 \quad \forall \  l \in \mL
\end{equation}
\end{lemma}
\begin{proof}
Suppose there is a link $ q $ for which at NE $ \sum_{k \in \mN^q} \alpha^q_k {x}_k < c^q $ and $ p^q > 0 $.
Again, we will show that deviation makes any agent $ j \in \mN^q $ better off.
Take the deviation $ {p_j^{\prime}}^q = p^q - \epsilon > 0 $ and no deviation in the demand.
We can write the difference after and before deviation as
\begin{gather}
\Delta \hat{u}_j = - ({p^{\prime}_j}^q - p^q)^2 - \eta \,  p^q({p^{\prime}_j}^q - p^q)(c^q - \sum_{k \in \mN^q} \alpha^q_k {x}_k ) \\
{} + \eta \, p^q(p^q - p^q)(c^q - \sum_{k \in \mN^q} \alpha^q_k {x}_k )   + (p^q - p^q)^2 \\
\Rightarrow \quad \Delta \hat{u}_j = - (-\epsilon)^2 -  \eta \, p^q(-\epsilon)(c^q - \sum_{k \in \mN^q} \alpha^q_k {x}_k ) + 0 + 0 = \epsilon \big( -\epsilon + \eta \, p^q(c^q - \sum_{k \in \mN^q} \alpha^q_k {x}_k ) \big)
\end{gather}
So if we take $ \epsilon $ such that
\begin{equation}
\min \Big \{ \underbrace{\eta \,  p^q}_{> 0} \underbrace{(c^q - \sum_{k \in \mN^q} \alpha^q_k {x}_k)}_{> 0},~~ \underbrace{p^q}_{> 0} \Big\} > \epsilon > 0
\end{equation}
then $ \Delta \hat{u}_j > 0 $. Taking such an $ \epsilon $ is possible because LHS above is positive.
\end{proof}

\begin{lemma}[Stationarity] \label{lemStat}
At any NE $ {s} = ({\novec{y}}, \novec{p}) $ of game $ \mathcal{G} $, and corresponding allocation, $ {x} $, we have
\begin{gather}
v_i^{\prime}({x}_i) = \sum_{l \in \mL_i} p^l \alpha^l_i \quad \forall~~ i \in \mN \quad \text{if} \quad {x}_i > 0 \\
v_i^{\prime}({x}_i) \le \sum_{l \in \mL_i} p^l \alpha^l_i \quad \forall~~ i \in \mN \quad \text{if} \quad {x}_i = 0
\end{gather}
\end{lemma}
\begin{proof}
At any NE $s$, agent $ i $'s utility
$ \hat{u}_i(s^{\prime}_i,s_{-i}) = v_i(h_{x,i}(s^{\prime}_i,s_{-i})) - h_{t,i}(s^{\prime}_i,s_{-i}) $
as a function of his message $ s_i^{\prime} = (y_i^{\prime}, \novec{p}_i^{\prime}) $, with $ s_{-i} $ fixed, should have a global maximum at $ {s}_i = ({y}_i, \novec{p}) $.
This would mean that if this function was differentiable w.r.t. $y_i^{\prime}$ at $s_i$, the partial derivatives w.r.t. $ y_i^{\prime} $ at $ s_i $ should be $ 0 $.
However, since our allocation dilates/shrinks demand vector $ \novec{y}^{\prime} $ on to the feasible region, it could be the case
that increasing and decreasing $ y_i^{\prime} $ gives allocations lying on different hyperplanes, meaning that the transformation
from $ \novec{y}^{\prime} $ to $ {x}^{\prime} $ is different on both sides of $ {y}_i $ and therefore  $ \hat{u}_i $ may not be differentiable
w.r.t $ y_i^{\prime} $ at $ {y}_i $.
The important thing here however is to notice that right and left derivatives exist, it's just that they may not be equal.  Hence we can take derivatives on both sides of $ {y}_i $ as (noting that derivative of the other two terms in utility will be zero due Lemma \ref{lemcmp})
\begin{align}
\label{EQder} \frac{\partial \hat{u}_i}{\partial y_i^{\prime}} \Big \vert_{y_i^{\prime} \downarrow {y}_i} = \left( v_i^{\prime}({x}_i) - \sum_{l \in \mL_i} p^l \alpha^l_i  \right) \frac{\partial x_i^{\prime}}{\partial y_i^{\prime}} \Big \vert_{y_i^{\prime} \downarrow {y}_i}, \quad  \frac{\partial \hat{u}_i}{\partial y_i^{\prime}} \Big \vert_{y_i^{\prime} \uparrow {y}_i} = \left( v_i^{\prime}({x}_i) - \sum_{l \in \mL_i} p^l \alpha^l_i  \right) \frac{\partial x_i^{\prime}}{\partial y_i^{\prime}} \Big \vert_{y_i^{\prime} \uparrow {y}_i}.   
\end{align}

We will first show that $ \partial x_i / \partial y_i $ term above (for either equation) is always positive. If $ y = 0 $ then clearly this is true, because if agent $ i $ demands $ y_i = \epsilon > 0 $ while $ y_{-i} = 0 $ then $ x_i > 0 $ (in fact the allocation is indeed differentiable at $ y = 0 $).
If $ y \ne 0 $ then from \eqref{EQx}, we can write
\begin{equation}
\beta \coloneqq \frac{\partial x_i}{\partial y_i} = \frac{\partial( r y_i)}{\partial y_i} = r + y_i \frac{\partial r}{\partial y_i} = r^q + y_i \frac{\partial r^q}{\partial y_i}
\end{equation}
where $ r = r^q $. 

From here on there are 3 cases: (A) $ i \notin \mN^q $; (B1) $ i \in \mN^q $  \&  $ \vert S^q(y) \vert \ge 2  $ and (B2) $ i \in \mN^q $  \&  $ \vert S^q(y) \vert = 1  $. (Note that $ \vert S^q(y) \vert = 0 $ isn't possible if $ r = r^q $ and $ y \ne 0 $)

(A) Here clearly $ \partial r^q / \partial y_i = 0 $, and this makes $ \beta = r^q > 0 $.

(B1) Here we can calculate $ \beta $ as
\begin{equation}
\beta = \frac{(r^q)^2}{c^q} \sum_{j \in \mN^q \backslash \{i\}} \alpha_j^q y_j
\end{equation}
which is positive because $ \vert S^l(y) \vert \ge 2 $, so there is at least one positive term in the summation.

(B2) In this case we could have $ S^q(y) = \{j\} \ne \{i\} $ or $ S^q(y) = \{i\} $. In the first case, argument is same as (A). So take $ S^q(y) = \{i\} $ and $ r^q = c^q/(\alpha_{i}^q y_i) - f^q(y_i) $.
\begin{equation}
\beta = -f^q(y_i) - y_i \frac{d {} f^q(y_i)}{d y_i} = - \frac{d ~ (y_i f^q(y_i))}{d y_i} = \frac{c^q}{\alpha_{i}^q (y_i + 1)^2}
\end{equation}
So we have $ \beta > 0 $ in all cases.

Now there are two further cases, the first term on RHS in both equations in \eqref{EQder} is positive or negative. If it's positive, then we can see from the first equation in \eqref{EQder} that by increasing $ y_i^{\prime} $ from $ {y}_i $ (and therefore $ x_i^{\prime} $ from $ {x}_i $) agent $ i $ can increase his pay-off, which contradicts equilibrium. Now similarly consider the first term in \eqref{EQder} to be negative, then from the second equation in \eqref{EQder}, agent $ i $ can reduce $ y_i^{\prime} $  from $ y_i $ to get a better pay-off. But the downward deviation in $ y_i^{\prime} $ is only possible if $ {y}_i > 0 $ ($ \Leftrightarrow~~ {x}_i > 0 $). So we conclude that
\begin{gather}
v_i^{\prime}({x}_i) = \sum_{l \in \mL_i} p^l \alpha^l_i \quad \forall~~ i \in \mN \quad \text{if} \quad {x}_i > 0 \\
v_i^{\prime}({x}_i) \le \sum_{l \in \mL_i} p^l \alpha^l_i \quad \forall~~ i \in \mN \quad \text{if} \quad {x}_i = 0.
\end{gather}
\end{proof}

Collecting the results of the above lemmas, we see that every NE satisfies the KKT conditions of the \eqref{CP}.
This means we now have necessary conditions on the NE up to the point of having unique allocation. In the next Lemma we will verify the existence of the equilibria that we have claimed.

\begin{lemma}[Existence] \label{lemexis}
For the game $ \mathcal{G} $, there exists equilibria $ s = (y,P) $, where price vectors are same
for each agent i.e. $ s = (y,p) $ and corresponding allocation-price pair $ (x,p) $ satisfy KKT conditions
with $\lambda =p$.
\end{lemma}
\begin{proof}
The proof is completed in two parts. Firstly we will check that for every $ x $ that can be a possible solution to \eqref{CP} (while satisfying the assumptions, specifically (A4))
there is indeed at least one $ y \in \mathbb{R}_{+}^N $ such that the allocation corresponding to $ y $ is $ x $.
(it is straightforward that this is true for prices and Lagrange multipliers $ \lambda $).
Secondly we will check that for the claimed NE $ s=(y,P) $, there are no unilateral deviation that are profitable.

In lieu of (A4), the optimal $ x^{\star} $ is such that $ \vert S^l(x^{\star}) \vert \ge 2 $ for all links; also it is clear that $ x^{\star} $ is on the boundary of the feasible region defined by \eqref{C1} and \eqref{C2}. For this, any vector $ y $ which is a scalar multiple of $ x^{\star} $ (and in particular $ y = x^{\star} $)
will give allocation $ x^{\star} $. This completes the first part of the proof.

Now we will check for profitable deviations. For this we want to show that at any action profile where corresponding allocation is $ x^\star $ and prices are equal and equal to $ \lambda^\star $, is a NE. Note that by construction, $ \hat{u}_i(s) $ is continuous w.r.t. $ s_i $.
Our approach would be to first characterise all local extrema.
As it turns out all local extrema can be characterised by $ \nabla \hat{u}_i = 0 $ $ \forall \ i $ (this isn't straightforward here since $ \hat{u}_i $ is only piecewise differentiable).
We then show that the set of points of local extrema can easily be seen to be precisely the set of points that satisfy the KKT conditions, i.e., where allocation is $ x^\star $ and prices are $ \lambda^\star $.
At the third step we show the Hessian for $ \hat{u}_i(s) $ w.r.t. $ s_i $ to be negative definite at all local extrema (for any agent $ i $). With this we know that all extremum points for $ \hat{u}_i $ w.r.t. $ s_i $ are local maxima, which can only happen if there is one maximum.
With the above uniqueness of maxima what we end up proving is that $( x^\star , \lambda^\star )$
is global optimum for every agent $i$ if others also play the same and thus a NE.

First we will show that without the gradient being zero, there cannot be a local extremum. Gradient of the utility function is
\begin{equation}
\nabla \hat{u}_i = \left(\frac{\partial \hat{u}_i}{\partial y_i}, \left(\frac{\partial \hat{u}_i}{\partial p_i^l}\right)_{l \in \mL_i} \right)
\end{equation}
Components in the gradient are
\begin{gather} \label{EQyder}
\frac{\partial \hat{u}_i}{\partial y_i} = \left(v_i^{\prime}(x_i) - \sum_{l \in \mL_i} \alpha^l_i \bar{p}_{-i}^l\right) \left(\frac{\partial x_i}{\partial y_i}\right) - \eta \sum_{l \in \mL_i} \bar{p}_{-i}^l (p_i^l - \bar{p}_{-i}^l) \left(-\sum_{j \in \mN^l} \alpha^l_j \frac{\partial x_j}{\partial y_i} \right) \\
\label{EQpder}
\frac{\partial \hat{u}_i}{\partial p_i^l} = -2(p_i^l - \bar{p}_{-i}^l) - \eta \, \bar{p}_{-i}^l (c^l - \sum_{j \in \mN^l} \alpha^l_j x_j) \quad \forall~~l \in \mL_i
\end{gather}
Since $ \hat{u}_i $ is differentiable w.r.t. $ p_i^l $ (at all points), that component of the gradient indeed has to be $ 0 $ at a local extremum. Which implies equal prices and complimentary slackness properties (using arguments from respective proofs) and we can write
\begin{equation}
\frac{\partial \hat{u}_i}{\partial y_i} = \left(v_i^{\prime}(x_i) - \sum_{l \in \mL_i} \alpha^l_i {p}^l\right) \left(\frac{\partial x_i}{\partial y_i}\right)
\end{equation}
Note that in the proof of Lemma~\ref{lemStat}, we have shown that $ \beta \coloneqq \frac{\partial x_i}{\partial y_i} > 0 $ always. So at the points of non-differentiability, $ \beta $ will have a jump discontinuity, however it will be positive on either side. It is then clear that without $ v_i^\prime(x_i) = \sum_{l \in \mL_i} \alpha_i^l p^l $, there cannot be a local extremum.


Second order partial derivatives are
\begin{align}
u_{pp} \coloneqq \frac{\partial^2 \hat{u}_i}{\partial p_i^{l}\partial p_i^{l}} = -2  \qquad
u_{lk} \coloneqq \frac{\partial^2 \hat{u}_i}{\partial p_i^{l} \partial p_i^k } = 0  \qquad
u_{py} \coloneqq \frac{\partial^2 \hat{u}_i}{\partial p_i^{l} \partial y_i} = \eta \, \bar{p}_{-i}^l \left(\sum_{j \in \mN^l } \alpha^l_j \frac{\partial x_j}{\partial y_i} \right)
\end{align} \begin{align}
u_{yy} \coloneqq \frac{\partial^2 \hat{u}_i}{\partial y_i \partial y_i} = \left(v_i^{\prime}(x_i) - \sum_{l \in \mL_i} \alpha^l_i \bar{p}_{-i}^l\right) \left(\frac{\partial^2 x_i}{\partial y_i^2}\right) + v_i^{\prime \prime}(x_i) \left(\frac{\partial x_i}{\partial y_i}\right)^2 \\
{} - \eta \sum_{l \in \mL_i} \bar{p}_{-i}^l (p_i^l - \bar{p}_{-i}^l) \left(-\sum_{j \in \mN^l} \alpha^l_j \frac{\partial^2 x_j}{\partial y_i^2} \right).
\end{align}
These derivatives will give us a Hessian $ H $ of size $ (L_i+1) \times ( L_i + 1 ) $, where 1st row and column represent $ y_i $ and subsequent rows and columns represent $ p_i^l $'s for different $ l $'s. We want $ H $ (evaluated at any local extremum) to be negative definite.
Now, 1st and 3rd terms in $ u_{yy} $ are zero at local extrema (as argued above), and the 2nd term is strictly negative due to strict concavity of $ v_i $.
This along with $ u_{pp} = -2 $ tells us that all diagonal entries in $ H $ are negative. Also notice that because of $ u_{lk} = 0 $, all off-diagonal entries other than the ones in first row and column are zero.
Finally, note that due to assumption (A2), all prices are finite at local extremum and so $ u_{py} $ will be finite.
We will show that roots of the characteristic polynomial of $ H $ (i.e. its eigenvalues) all become negative if $ \eta $ is chosen sufficiently small.
Here again, we can use (A4) 
to justify working with $ r $ that has the form defined in the first sub-case of \eqref{EQA1}.

For this, we take a generic matrix $ A = \{a_{ij}\}$, which is similar in structure to $ H $ and has the same dependence on $ \vert y \vert $ as $ H $. So entries in $ A $ are
\begin{gather}
a_{11} = -\frac{a}{\vert y \vert^2} \qquad a_{ij} = a_{ji} = 0 \qquad \forall \quad i,j > 1,~~ i \ne j  \\
a_{ii} = -2  \qquad a_{1i} = a_{i1} = \eta \, \frac{b_{i-1}}{\vert y \vert} \qquad \forall \quad 2 \le i \le L_i+1
\end{gather}
where $ a > 0 $ (and we don't care about the sign of $ b_i $'s). The parameters $ a,b_i $ may not be completely independent of $ y $ but since the absolute value of $ y $ has been taken out of the scaling, their values are bounded. Magnitude of $ b_i $'s are bounded from above and $ a $ is bounded away from zero. Now we can explicitly calculate $ \vert A - \lambda I \vert $ and write the characteristic equation as
\begin{equation}
Q(\lambda) = \left(-\frac{a}{\vert y \vert^2} - \lambda \right) (-2-\lambda)^{L_i} + \eta^2 \, \frac{\sum_{i = 1}^{L_i} (-1)^{i} b_i^2}{\vert y \vert^2} (-2-\lambda)^{L_i-1}  = 0
\end{equation}
So $ -2 $ is a repeated eigenvalue, $ L_i - 1 $ times. The equation for the remaining two roots can be written as
\begin{equation}
\left(-\frac{a}{\vert y \vert^2} - \lambda \right) (-2-\lambda) + \eta^2 \frac{C}{\vert y \vert^2}  = 0
\end{equation}
Necessary and sufficient conditions for both roots of this quadratic to be negative are
\begin{equation}
\left(2 + \frac{a}{\vert y \vert^2}\right) > 0 \qquad \frac{2a}{\vert y \vert^2} + \eta^2 \frac{C}{\vert y \vert^2} > 0,
\end{equation}
first of which is always true, since $ a > 0 $. The second one can be ensured by making $ \eta $ small enough, since $ a $ is bounded away from zero and magnitude of $ C $ is bounded from above. 

Hence we have shown the Hessian $ H $ to be negative definite for $ \eta $ chosen to be small enough.
%
\end{proof}

Several comments are in order regarding the selection of the proportional allocation mechanism and in particular~\eqref{EQA1}.
If we use ``pure" proportional allocation i.e. same expression for $ r^l $ for $ \vert S^l(y) \vert \ge 2 $ and $ \le 1 $, then irrespective of optimal solution of \eqref{CP}, for game $ \mathcal{G} $ the ``stationarity'' property will not be satisfied for equilibria with $ \vert S^l(y) \vert \le 1$. Thus the mechanism will result in additional extraneous equilibria.
For this reason we tweak the expression for $ r^l $ when $ \vert S^l(y) \vert \le 1 $, so that we can eliminate these extraneous equilibria - irrespective of the solution of \eqref{CP}.
With this tweak in the expression for $ r^l $, all KKT conditions become necessary for all
equilibria regardless of the value of  $\vert S^l(y) \vert $.
This however creates a problem in the proof of existence of equilibria.
In particular, if $ x^{\star} $ was such that it had links where $ \vert S^l(x^{\star}) \vert = 1 $ then in our allocation this would require $ y $ at NE such that $ \vert S^l(y) \vert = 1 $.
In this case the $ r^l $ used would be lower than what the proportional allocation requires (see second sub-case in \eqref{EQA1}) and we actually would have the problem of possibly not having any $ y $ that creates $ x^{\star} $ as allocation.
Hence we have used (A4) to eliminate this case.

\begin{lemma}[Individual Rationality] \label{lem_ir}
At any NE $ s = (\novec{y}, \novec{p}) $ of $ \mathcal{G} $, with corresponding allocation $ x $, we have
\begin{align} \label{EQIRB}
u_i(x,t) &\ge u_i(0,0) \quad \forall~~ i \in \mN \\
\text{and} \qquad \sum_{i \in \mN} t_i &\ge 0  \qquad (WBB)
\end{align}
\end{lemma}
\begin{proof}
Because of Lemma \ref{lemcmp}, the only non-zero term in $ t_i $ at equilibrium is $ x_i \sum_{l \in \mL_i} \alpha^l_i p^l $, which is clearly non-negative. Hence $ \sum_{i \in \mN} t_i \ge 0 $ at equilibrium. This is the seller's individual rationality condition.

Now if  $ x_i = 0$ then we know from Lemma \ref{lemcmp} and  \eqref{EQt} that $ t_i = 0 $ and so \eqref{EQIRB} is evident. Now take $ x_i > 0 $ and define the function
\begin{equation}
f(z) = v_i(z) - z \sum_{l \in \mL_i} \alpha^l_i p^l.
\end{equation}
Note that $ f(0) = u_i(0,0) $ and $ f(x_i) = u_i(x,t) $, the utility at equilibrium. Since $ f^{\prime}(x_i) = 0 $ (Lemma \ref{lemStat}), we see that $ \forall $ $ 0 < y < x_i $, $ f^{\prime}(y) > 0 $ since $ f $ strictly concave (because of $ v_i $). This clearly tells us $ f(x_i) \ge f(0) $.
\end{proof}

Now that we have Lemmas characterising NE in the same way as KKT conditions (and individual rationality), we can compare them to prove Theorem \ref{thmain}.

\begin{proof}[Proof of Theorem \ref{thmain}]
We know that the four KKT conditions produce a unique solution $ x^{\star} $ (and corresponding $ \lambda^{\star} $). For the game $ \mathcal{G} $, from Lemmas \ref{lemPF}--\ref{lemStat} we can see that at any NE, allocation $ x $ and prices $ \novec{p} $ satisfy the same conditions as the four KKT conditions and hence they give a unique $ x =x^{\star} $, as long as (A4) is satisfied. So we have that the allocation is $ x^{\star} $ across all NE.
This combined with individual rationality Lemma~\ref{lem_ir}, proves Theorem \ref{thmain}.
\end{proof}

\section{A Mechanism with Strong Budget Balance} \label{secmechown}

We now turn our attention to problem \eqref{CP2}. So in this case we have the agents in $ \mN $, who are the owners and users that wish to allocate the good amongst themselves in a way that maximises $ \sum_{i\in\mN} u_i $. In this case one can now think of taxes as a way of facilitating efficient redistribution of the already available good. Since all payments are made amongst agents in $ \mN $ and we have quasi-linear utilities, this clearly tells us that $ \sum_{i\in\mN} t_i $ must be zero. This interpretation is slightly different from Section \ref{secmech}, where taxes were indeed payments made to the seller for provisioning of the good.

All of the above is required to be done again under the assumption of strategic users, which means the designer (who is a third party) still has the problem of information elicitation and moreover has to make sure that the wealth has to be redistributed in a way that we still get $ x^{\star} $ allocation at the all equilibria. Here we will say that the mechanism fully implements maximising social welfare allocation if in addition to the previous conditions, we also have SBB.

\subsection{Information assumptions} These are the same as Section \ref{secmech}.

\subsection*{} For creating a mechanism in this formulation, main difference with the previous section, is that we have to find a way of redistributing the total tax paid by all the agents. In the last section we saw that the total payment made at the equilibrium is
\begin{equation} \label{EQtotpay}
 B= \sum_{i \in \mN} \left( x_i \sum_{l \in \mL_i} \alpha^l_i p^l \right) = r \sum_{i \in \mN} \left( y_i \sum_{l \in \mL_i} \alpha^l_i p^l \right)
\end{equation}
since all other tax terms were zero at equilibrium.
We will redistribute taxes by modifying tax function for each agent in such a way that he only uses messages from other agents.
This has the advantage of keeping our equilibrium calculations in line with Section \ref{secmech}, since deviations by an agent wouldn't affect his utility through this additional term. In view of this, we can express $B$ as follows
\begin{equation} \label{EQALTREP}
B = r \sum_{i \in \mN} \left( \sum_{l \in \mL_i} \frac{p^l}{N^l - 1} \sum_{j \in \mN^l \backslash \{i\}} \alpha^l_j y_j \right),
\end{equation}
where each term of the outer summation depends only on demands of agents other than the $ i^{th} $ one.
This means that each term in the parenthesis (scaled by the factor $r$) can now be used as the desired additional tax for user $i$.
Observe, however,  that in our mechanism, each agent's demand affects the factor $r$ as well.
So, if all agents can agree on value of $ r $ then we can use that signal
to create the term that facilitates budget balance.

In lieu of this, our mechanism here works by asking for an additional signal $ \rho_i $ from every agent and imposing an additional tax of $ (\rho_i - r)^2 $, thereby essentially ensuring that all agents agree on the value of $ r $ (via $ \rho_i $'s) at equilibrium. Finally, we use $ \bar{\rho}_{-i} $ (cf. \eqref{EQRHO}) as a proxy for $ r $ in \eqref{EQALTREP} - just like we did with $ \bar{p}_{-i}^l $'s.

\subsection{Mechanism}
Now the actions sets $ \mS_i $ for agents will be $ \mathbb{R}_{+} \times \mathbb{R}_{+}^{L_i} \times \mathbb{R}_{+} $ and actions will look like $ s_i = (y_i, p_i, \rho_i) $.

The designer announces the contract $ h: \mS \rightarrow \mathbb{R}_{+}^N \times \mathbb{R}^N $ and asks each agent to submit their message $s_i= (y_i, p_i, \rho_i) $. Then he makes allocations and taxes based on the contract for each agent $ i \in \mN $ exactly as in the WBB case, with the only exception that the tax is now defined as
%
\begin{gather}\label{EQt2}
h_{t,i}(s) = t_i = \sum_{l \in \mL_i} t_i^l + \zeta (\rho_i - r)^2   \\
t_i^l =  x_i \alpha^l_i \bar{p}_{-i}^l  + (p_i^l -  \bar{p}_{-i}^l)^2  +  \eta \, \bar{p}_{-i}^l (p_i^l -  \bar{p}_{-i}^l)(c^l - \sum_{j \in \mN^l} \alpha^l_jx_j) -  \frac{ \bar{\rho}_{-i} \bar{p}_{-i}^l}{ N^l - 1} \sum_{j \in \mN^l \backslash \{i\} } \alpha^l_j y_j.
\end{gather}
where $ \zeta,\eta $ are small enough positive constants. Here $ \bar{p}_{-i}^l $ is as defined in \eqref{EQpm} and
\begin{gather} \label{EQRHO}
\bar{\rho}_{-i} \coloneqq \frac{1}{N-1} \sum_{j \in \mN \backslash \{i\}} \rho_j.
\end{gather}
Here we will call the corresponding game $ \mathcal{G}_0 $, for which utilities will be
\begin{equation}
\hat{u}_i(s) = v_i(x_i) - t_i = v_i(h_{x,i}(s)) - h_{t,i}(s) \qquad \forall~~i \in \mN
\end{equation}

We will now move on to results section and discuss the implications of the modifications there.

\subsection{Results}

With this new mechanism, we will again have full implementation (note that for individual rationality there is no seller here). The only term in $ \hat{u}_i $ that is affected by $ \rho_i $ is $ - \zeta (\rho_i - r)^2 $, so all the Lemmas from Section \ref{secmech} will go through with minor modifications and we will have our main result using the same line of argument as for Theorem \ref{thmain}. Note here that, terms in $ \hat{u}_i $ affected by $ p_i^l $'s are the same as before but for $ y_i $ there is a new term $ - \zeta (\rho_i - r)^2 $ which is affected by it.

\begin{theorem}[Full Implementation]\label{thmain2}
For game $ \mathcal{G}_0 $, there is a unique allocation, $ x $, corresponding to all NE. Moreover, $ x = x^{\star} $, the maximiser of \eqref{CP}, where individual rationality is satisfied for all agents in $ \mN $.
\end{theorem}


In addition to all the properties from Section \ref{secmech}, here we will characterise $ \rho_i $'s at equilibrium and then go on to show SBB at equilibrium.

Instead of proving the results from Section \ref{secmech} for this new mechanism, we will outline their proofs and only show the rigorous proofs for new properties.
\begin{itemize}
\item Primal Feasibility - Since allocation function is the same as before, this result holds here as well.

\item Equal Prices at equilibrium - This was proved by taking price deviations only and keeping other parameters of the signal constant, so the same argument works here as well (noting that no new price related terms have been added in the new mechanism).
\end{itemize}
Before moving on to other results, we will show common $ \rho_i $'s at equilibrium.

\begin{lemma} \label{lemcmr}
At any NE $ s = (y, p, \rho) $ of game $ \mathcal{G}_0 $, we have $ \rho_i = r $ $ ~\forall $ $ i \in \mN $.
\end{lemma}
\begin{proof}
Suppose not, i.e. assume $ \exists $ $ j \in \mN $ such that $ \rho_j \ne r $. In this case agent $ j $ can deviate with only changing $ \rho_j^{\prime} = r $ (which also means $ r $ is the same as before deviation, since demand $ y $ doesn't change). It's easy to see that this is a profitable deviation, since change in utility of agent $ j $ will be only through the term involving $ \rho_j $.
\begin{equation}
\Delta \hat{u}_j = - \zeta (\rho_j^{\prime} - r)^2 + \zeta (\rho_j - r)^2  = \zeta (\rho_j - r)^2 > 0
\end{equation}
\end{proof}
Note however that although $ \rho_i $ are same for all $ i $ at any equilibrium, that common value, $r$,  will be different across equilibria. This is obvious since magnitude of vector $ y $ changes across equilibria.

Now we move on with properties from Section \ref{secmech}.
\begin{itemize}
\item Dual Feasibility - This is obvious here as well.

\item Complimentary Slackness - This was proved by taking only price deviations and hence the same argument works here as well.



\item Stationarity - Now the additional term in the derivative here will be
\begin{equation}
\frac{\partial \hat{u}_i}{\partial y_i^{\prime}} \Big \vert_{new} = \underbrace{\frac{\partial \hat{u}_i}{\partial y_i^{\prime}} \Big \vert_{old}}_{T_1} - \underbrace{2 \zeta (\rho_i^{\prime} - r^{\prime}) \left(-\frac{\partial r^{\prime}}{\partial y_i^{\prime}}\right)}_{T_2}
\end{equation}
So we claim as before that if $ T_1 $ is positive, we can increase $ y_i^{\prime} $ from $ y_i $ to be better-off. Here however we would have to make sure that agent $ i $ deviates with $ \rho_i^{\prime} $ simultaneously to make it equal to $ r^{\prime} $, so that the contribution of the $ T_2 $ term to the derivative is zero. The only thing left to notice here is that the change in $ \rho_i^{\prime} $ is such that not only the term $ T_2 $ is zero but also that the contribution of term $ - \zeta (\rho_i^{\prime} - r^{\prime})^2 $ to the utility is zero before and after deviation - so this deviation doesn't change other partial derivatives. Similar argument also works when $ T_1 $ is negative and we get the stationarity property here as well.
\end{itemize}
With this we will have unique allocation at every equilibria, since solution to KKT is unique (as far as allocation is concerned). This unique allocation will be $ x^{\star} $; also the prices will be $ \lambda^{\star} $, same as before.

Now we verify the existence of equilibria. The arguments here will be similar to the ones in the proof of Lemma \ref{lemexis}. First order conditions can again be shown to be satisfied, the only difference is that here we will also use $ \rho_i = r $ at local extremum. The Hessian $ H $ here, for agent $ i $, will be of order $ ( L_i +2) \times (L_i +2) $ where 1st, 2nd row and column represent $ y_i $, $ \rho_i $ respectively whereas the remaining rows and columns represent $ p_i^l $'s. The generic matrix $ A = \{a_{ij}\} $ for $ H $ will then be
\begin{gather}
a_{11} = -\frac{a}{\vert y \vert^2} - \zeta \frac{d}{\vert y \vert^4} \qquad a_{12} = a_{21} = - \zeta \frac{e}{\vert y \vert^2} \qquad  a_{ij} = a_{ji} = 0 \qquad \forall \quad i,j > 1,~~ i \ne j  \\
a_{22} = -2 \qquad a_{ii} = -2  \qquad a_{1i} = a_{i1} = \eta  \frac{b_{i-1}}{\vert y \vert} \qquad \forall \quad 3 \le i \le L_i + 2
\end{gather}
where $ a,d,e > 0 $. Writing the characteristic equation we will again get that $ -2 $ is a repeated eigenvalue, $ L_i $ times. And the equation for remaining two roots is
\begin{equation}
\lambda^2 + \lambda \left(2 + \frac{a}{\vert y \vert^2} + \zeta \frac{d}{\vert y \vert^4}\right) + \left(\frac{2a}{\vert y \vert^2} + \zeta\frac{2d}{\vert y \vert^4} - \zeta^2 \frac{e^2}{\vert y \vert^4} + \eta^2 \frac{C}{\vert y \vert^2}\right) = 0
\end{equation}
Necessary and sufficient conditions for the roots of above quadratic to be negative are again that coefficient of $ \lambda $ and the constant term are both positive. Coefficient of $ \lambda $ is clearly positive, and the constant term can also be made positive by choosing $ \zeta,\eta $ small enough, irrespective of sign of $ C $. Hence here also we get NE for all $ y $ (along a fixed direction) for $ \zeta,\eta $ chosen to be small enough.

\begin{itemize}
\item Individual Rationality - This is obvious in here because we are only redistributing money from the previous case, so if the mechanism there was individually rational it will be here too.
\end{itemize}

\begin{lemma}[Strong Budget Balance]
At any NE $ s = (y,p, \rho) $ of game $ \mathcal{G}_0$, with corresponding taxes $ (t_i)_{i \in \mN} $, we have $ \sum_{i \in \mN} t_i = 0 $.
\end{lemma}
\begin{proof}
We now know that price vectors are equal at equilibrium for all agents and so we can write
\begin{gather}
\sum_{i\in\mN} t_i = \sum_{i\in\mN} {x}_i \left (\sum_{l \in \mL_i} \alpha^l_i p^l \right) - r \sum_{l \in \mL_i} \frac{{p}^l}{ N^l - 1} \sum_{j \in \mN^l \backslash \{i\}} \alpha^l_j y_j\\
\label{EQsbb}
\Rightarrow \quad \sum_{i\in\mN} t_i = \sum_{i\in\mN} \sum_{l\in\mL_i} \left( {x}_i \alpha^l_i p^l - \frac{p^l}{ N^l -1} \sum_{j \in \mN^l \backslash \{i\}} \alpha^l_j {x}_j \right)
\end{gather}
Consider the coefficient of $ {x}_k $ for any agent $ k $ in the above expression
\begin{align}
\sum_{l\in\mL_k} \left( \alpha^l_k p^l - \sum_{i \in \mN^l \backslash \{k\}} \frac{p^l}{N^l - 1} \alpha^l_k \right)
 &= \sum_{l\in\mL_k} \left(  \alpha^l_k p^l - \frac{p^l \alpha^l_k}{N^l - 1}  \left (N^l - 1 \right )   \right) \\
 &=0,
\end{align}
which proves the claim.
\end{proof}

\begin{proof}[Proof of Theorem \ref{thmain2}]
So by the preceding properties, we get allocation $ x^{\star} $, prices $ \lambda^{\star} $ at all equilibria. Then SBB and individual rationality give us the desired full implementation.
\end{proof}



\section{Discussion and Generalizations} \label{secgen}

\subsection*{Relevant Literature}
The problem considered in \cite{demos} is essentially equivalent to ours (with the additional property of
SBB on and off equilibrium and the relaxed assumptions (A3)-(A4)).
However, as it turns out the claim of implementation made in~\cite{demos} (property (P1)) is not valid.
In particular the proof of~\cite[Theorem~5]{demos}  (specifically eq.~(64)) is incorrect,
since the utilities need not have zero derivatives at equilibrium, since they are
discontinuous at equilibrium and the only allowable deviations of $x_i$ are downwards deviations.
Unfortunately, this is a fundamental problem with the proposed mechanism in~\cite{demos} and not
merely a fixable error in the proof. In addition, there is no set of relaxed assumptions for which
the proposed mechanism can implement the solution of the~\eqref{CP}.
Intuitively, that mechanism fails to achieve the claim of implementation because of using hard constraints for
ensuring feasibility of allocation: when the demanded allocation is not feasible, a large penalty
is imposed on the agents.
This approach creates discontinuities of the utility functions at the boundaries of the achievable
region and thus renders invalid any attempt to link
the corresponding NE with the KKT conditions of the corresponding centralized problem.
Indeed, one of the main contributions of our work in this report is embedding the
constraints within the mechanism in an implicit way, such that the allocations are always
feasible (on and off equilibrium), and  are continuous and (piecewise) differentiable with
respect to the demands.
The authors of~\cite{demos} have suggested an alternative approach of overcoming
these difficulties in~\cite{demoscorrection}.

The allocation that we use in our work can be referred to as proportional allocation, since it gives each
agent an allocation which is proportional to their demand albeit weighed down by the total demand - in a way
that respects the constraints.
This idea was used in the case of 1 link ($ L=1 $) in \cite{hajek,basar}, with a different payment scheme.
Both these papers achieve partial (not full) implementation, in the sense that there exists at least one
NE which gives the required allocation.

\subsection*{Strong Budget Balance off-equilibrium}
In this work, we do not view SBB off equilibrium as an important property of a mechanism.
In fact one may suggest that if any property should hold on and off equilibrium that would be
feasibility, since otherwide the network wouldn't be able to operate at all if equilibrium is not reached.
However, in the following we sketch a modification of the proposed mechanism that results in arbitrarily close to SBB even off equilibrium.
In Section \ref{secmechown}, we use $ \bar{\rho}_{-i} $'s simply as a way to get SBB at equilibrium. Here $ \bar{\rho}_{-i} $ was used as a proxy for $ r $, since we knew that at equilibrium we will have $ \bar{\rho}_{-i} = r $. We could, in addition to this, also use $ \bar{\rho}_{-i} $ as a proxy for $ r $ in the allocation i.e. $ x_i = \bar{\rho}_{-i} y_i $. Although we won't have feasibility of allocation off-equilibrium, this will ensure that the first term (payment) and fifth term in tax function cancel out when we sum over all agents - on or off-equilibrium. This will give us something close to SBB at all points in the message space $ \mS $ and not just at equilibria - for this all we have to notice is that in \eqref{EQt2} we could introduce any positive constant in front of terms 2,3 and 4 and all the results would still go through. So by making that constant small enough we could restrict the contribution of those terms to $ \sum_{i \in \mN} t_i $, which we couldn't do with terms 1 and 5
since term 1
compares with $ v_i $, for which we do not know the scaling and term 5 is introduced to cancel out term 1 when we sum over all agents.

\bibliographystyle{IEEEtran}
\bibliography{abhinav}

\end{spacing}
\end{document}